\documentclass[12pt]{article}

\usepackage{latexsym}
\usepackage{epsfig,graphics}
\usepackage{amsmath,amssymb,amsthm}
\usepackage{graphics,epsfig,calc}
\usepackage{upref}
 \newcommand{\beqn}{\begin{eqnarray}}
 \newcommand{\eeqn}{\end{eqnarray}}
 \newcommand{\be}{\begin{equation}}
 \newcommand{\ee}{\end{equation}}
 \newcommand{\ba}{\begin{array}}
 \newcommand{\ea}{\end{array}}

 \newcommand{\pa}{\partial}
 \newcommand{\re}{\ref}
 \newcommand{\ci}{\cite}
 \newcommand{\ds}{\displaystyle}
 \newcommand{\la}{\label}
 
 \newcommand{\rRe}{{\rm Re\5}}
 \newcommand{\fr}{\frac}

\newcommand{\ov}{\overline}

\newcommand{\ti}{\tilde}

\newcommand{\cF}{{\cal F}}
\newcommand{\cA}{{\cal A}}

\newcommand{\cG}{{\cal G}}
\newcommand{\cH}{{\cal H}}

\newcommand{\cL}{{\cal L}}
\newcommand{\cR}{{\cal R}}
\newcommand{\cP}{{\cal P}}
\newcommand{\cO}{{\cal O}}

\newcommand{\cV}{{\cal V}}
\newcommand{\ve}{\varepsilon}

\newcommand{\de}{\delta}

\newcommand{\ga}{\gamma}
\newcommand{\Ga}{\Gamma}
\newcommand{\si}{\sigma}
\newcommand{\om}{\omega}

\newcommand{\na}{\nabla}
\newcommand{\Si}{\Sigma}
\newcommand{\lam}{\lambda}

\newcommand{\5}{{\hspace{0.5mm}}}
\newcommand{\N}{\mathbb{N}}
\newcommand{\R}{\mathbb{R}}

\newcommand{\C}{\mathbb{C}}

\newtheorem{theorem}{Theorem}[section]

\newtheorem{defin}[theorem]{Definition}

\newtheorem{lemma}[theorem]{Lemma}

\newtheorem{cor}[theorem]{Corollary}
\newtheorem{pro}[theorem]{Proposition}
\newcommand{\bp}{\begin{pro}}
\newcommand{\ep}{\end{pro}}

\newcommand{\const}{\mathop{\rm const}\nolimits}

\begin{document}
\begin{titlepage}
\bigskip\bigskip\bigskip

\begin{center}
{\Large\bf On long-time  decay for modified Klein-Gordon equation}
\vspace{1cm}
\\
{\large E.~A.~Kopylova}
\footnote{Supported partly by  FWF,
DFG and RFBR grants.}\\
{\it Institute for Information Transmission Problems RAS\\
B.Karetnyi 19, Moscow 101447,GSP-4, Russia}\\
e-mail:~elena.kopylova@univie.ac.at
\end{center}

\date{}

\vspace{0.5cm}
\begin{abstract}
\noindent
We obtain a dispersive long-time decay in weighted energy norms for solutions
of the  Klein-Gordon equation in a moving frame. The decay extends the
results of Jensen, Kato and Murata 
for the equations of the Schr\"odinger type. We modify the
approach  to make it applicable to relativistic equations.
\smallskip

\smallskip

\end{abstract}

\end{titlepage}
\setcounter{equation}{0}
\section{Introduction}
In this paper, we establish a dispersive long time decay in weighted energy norms
for the solutions to  1D Klein-Gordon equation in a moving frame with
the velocity $v$
\be\la{KGEr} \dot \Psi(t)=\cA \Psi(t)
\ee
where
$$
\Psi(t)=\left(\begin{array}{c}
  \psi(t)
  \\
  \pi(t)
  \end{array}\right),\quad
 \cA =\left(\begin{array}{cc}
  v\nabla                          &   1
  \\
 \Delta-m^2-V  &   v\nabla
  \end{array}\right),\quad\na=\fr{d}{dx},\quad\Delta=\fr{d^2}{dx^2}
$$
with $m>0$, and $|v|<1$.
For $s,\si\in\R$, we denote by $H^s_\si=H^s_\si (\R)$
the weighted  Agmon-Sobolev spaces \ci{A}, with the finite norms
$$
  \Vert\psi\Vert_{H^s_\si}=\Vert\langle x
  \rangle^\si\langle\na\rangle^s\psi\Vert_{L^2(\R)}<\infty,
  \quad\quad \langle x\rangle=(1+|x|^2)^{1/2}
$$
Denote $L^2_\si=H^0_\si$. We assume that $V(x)$ is a real function, and
\be \label{V}
    |V(x)|+|V'(x)|\le C\langle x\rangle^{-\beta},~~~~~x\in\R
\ee
for some $\beta>5$. Then the multiplication by $V(x)$
is bounded operator $H^1_s \to H^{1}_{s+\beta}$ for any $s\in\R$.

We consider the ``nonsingular case'' in the terminology of
\cite{M}, when the truncated resolvent of the
operator $-\Delta+\ga^2V(x)$, $\ga=1/\sqrt{1-v^2}$ is bounded at the edge point
$\zeta=0$  of the continuous spectrum. In other words,
\be\la{SC} 
{\it the~ point}~\zeta=0~
{\it is~ neither~ eigenvalue~ nor~ resonance~ for~ the~
operator}~-\Delta+\ga^2V(x)
\ee
\smallskip
By definition (see \cite[page 18]{M}) the point $\zeta=0$ is the
 resonance if there exists a nonzero solution $\psi\in L^2_{-1/2-0}\setminus L^2$
to the equation $(-\Delta+\ga^2V(x))\psi=0$.
\begin{defin}\la{def2}
$\cF _{\si}$ is the  complex  Hilbert space $H^1_{\si}\oplus H^0_{\si}$
of vector-functions $\Psi =(\psi ,\pi )$ with the norm
$$
   \Vert \,\Psi\Vert_{\cF _{\si}}=
   \Vert\psi\Vert_{H^1_\si} +\Vert\pi \Vert_{H^0_\si}<\infty
$$
\end{defin}
Our main result is the following long time decay of the solutions
to (\re{KGEr}): in the nonsingular case, the asymptotics hold
\be\label{full}
 \Vert\cP_c\Psi(t)\Vert_
 {\cF_{-\si}}=\cO (|t|^{-3/2}),\quad t\to\pm\infty
\ee
for initial data $\Psi_0=\Psi(0)\in\cF _\si$ with $\sigma>5/2$, where  $\cP_c$
is a Riesz projection onto the continuous spectrum of the operator $\cA$.
The decay is desirable for the study of asymptotic stability and scattering
for the solutions to nonlinear hyperbolic equations.

Let us comment on previous results in this direction.
The decay of type (\re{full}) in weighted norms has
been established first by Jensen and Kato \ci{jeka} for the Schr\"odinger
equation in the dimension $n=3$. The result has been extended to all
other dimensions by Jensen and Nenciu \ci{je1,je2,jene}, and to more
general PDEs of the Schr\"odinger type by Murata \ci{M}.

The Jensen-Kato-Murata approach is not applicable directly
to the relativistic equations. The difference reflects distinct character of wave
propagation in the relativistic and nonrelativistic equations
(see the discussion in \ci[Introduction]{1dkg}).

In \ci{1dkg} the decay of type (\re{full}) in the weighted energy norms
has been proved  for the 1D Klein-Gordon equation with $v=0$.
The approach  develops the Jensen-Kato-Murata techniques to make it applicable to the
relativistic equations. Namely, we apply the finite Born series and convolution.
Here we extend the result  \ci{1dkg} to the case $v\not=0$.

Our paper is organized as follows.
In Section \ref{fKG} we obtain the time decay for the solution to the free
modified Klein-Gordon equation and state the spectral properties of the
free resolvent..
In Section \ref{pKG} we obtain spectral properties of the perturbed resolvent
and prove the decay \eqref{full}.

\setcounter{equation}{0}
\section{Free equation}\la{fKG}
Here we consider the free  equation with zero potential $V(x)=0$:
\be\la{KGE1} \dot \Psi(t)=\cA_0 \Psi(t)
\ee
where
$$
\cA_0 =\left(\begin{array}{cc}
  v\nabla                          &   1
  \\
 \Delta-m^2  &   v\nabla
  \end{array}\right)
$$
\subsection{Spectral properties}
For $t>0$ and $\Psi_0=\Psi(0)\in\cF_0$,
the solution $\Psi(t)$ to (\re{KGE1})
admits the spectral Fourier-Laplace representation
\be\la{Gint}
  \theta(t)\Psi(t)=-\fr 1{2\pi}\int\limits_{\R}e^{(i\om+\ve) t}
  \cR _0(i\om+\ve)\Psi_0~d\om,~~~~t\in\R
\ee
with any $\ve>0$, where $\theta(t)$ is the Heaviside function,
$\cR _0(\lam)=({\cal A}_0-\lam)^{-1}$ for $\rRe\lam>0$
is the resolvent of the operator ${\cal A}_0$.
The representation follows from the stationary equation
$\lam\ti\Psi^+(\lam)=\cA_0\ti\Psi^+(\lam)+\Psi_0$
for the Fourier-Laplace transform
$\ti\Psi^+(\lam):=\ds\int_\R \theta(t)e^{-\lam t}\Psi(t)dt$, $\rRe\lam>0$.
The solution  $\Psi(t)$ is  continuous bounded  function of
$t\in\R$ with the values in $\cF_0$
by the energy conservation for the equation (\re{KGE1}).
Hence, $\ti\Psi^+(\lam)=-\cR_0(\lam)\Psi_0$
is analytic function in $\rRe\lam>0$ with the values in $\cF_0$, and
bounded for $\rRe\lam>\ve$. Therefore, the integral (\re{Gint})
converges in the sense of distributions of $t\in\R$ with the values in $\cF_0$.
Similarly to (\re{Gint}),
\be\la{Gints}
  \theta(-t)\Psi(t)=\fr 1{2\pi}\int\limits_{\R}e^{(i\om-\ve) t}
  \cR _0(i\om-\ve)\Psi_0~d\om,~~~~t\in\R
\ee
Let us calculate  the resolvent $\cR _0(\lam)$. We have
$$
\cR_0(\lam)=(\cA_0-\lam)^{-1}=
\left(\ba{cc}
v\nabla-\lam             & 1 \\
\Delta-m^2               & v\nabla-\lam
\ea\right)^{-1},\quad\rRe\lam>0
$$
In the  Fourier space we  obtain
$$
\left(\ba{cc}
-(ivk+\lam) & 1 \\
-(k^2+m^2) & -(ivk+\lam)
\ea\right)^{-1}=[(ivk+\lam)^2+k^2+m^2]^{-1}\left(
\ba{cc}
-(ivk+\lam) & -1 \\
k^2+m^2 & -(ivk+\lam)
\ea\right)
$$
Taking the inverse Fourier transform, we obtain the resolvent
\be\la{Green}
\cR_0(\lam)=\left(\ba{cc}
v\nabla-\lam   &                -1 \\
-\Delta+m^2    & v\nabla-\lam
\ea\right)R_0(\lam)=\left(\ba{cc}
(v\nabla-\lam)R_0(\lam)      & -R_0(\lam) \\
1-(v\nabla-\lam)^2R_0(\lam)  & (v\nabla-\lam)R_0(\lam)
\ea\right)
\ee
where $R_0(\lam)$
is the  operator with the integral kernel
\be\la{dete}
R_0(\lam,x,y)=F^{-1}_{k\to x-y}\ds\fr{1}{k^2+m^2+(ivk+\lam)^2},~~~x,y\in\R
\ee
which is well defined since the denominator in (\re{dete}) does not vanish
for $\rRe\lam>0$.
Denote  $\cH_0=-(1-v^2)\Delta+m^2=-\fr 1{\ga^2}\Delta+m^2$. Since
\be\la{cH0}
(\cH_0+\lam^2-2v\lam\na)\psi(x)=e^{-\ga^2v\lam x}(\cH_0+\ga^2\lam^2)
e^{\ga^2v\lam x}\psi(x)
\ee
we have
\be\la{RtilR}
R_0(\lam)=(\cH_0+\lam^2-2v\lam\na)^{-1}
=e^{-\ga^2v\lam x}\ga^2\tilde R_0(\ga^2m^2+\ga^4\lam^2)e^{\ga^2v\lam y}
\ee
where
$$
\tilde R_0(\zeta)=\!=\!(-\Delta+\zeta)^{-1}
={\rm Op}\Big[\frac{e^{-\sqrt{\zeta}|z|}}{2\sqrt{\zeta}}\Big]
$$
is the  Schr\"odinger resolvent. Finally, 
\be\la{glam}
R_0(\lam,x,y)=\fr{e^{-\ga^2(\sqrt{\lam^2-\mu^2}\,|x-y|+v\lam (x-y))}}
{2\sqrt{\lam^2-\mu^2}},~~ \mu\!=\!\fr{im}{\ga}
\ee
Denote
$\Gamma:=(-i\infty,-\mu,)\cup (\mu,~i\infty)$.
We  choose $\rRe\sqrt{\lam^2-\mu^2}>0$ for $\lam\in\C\setminus\ov\Gamma$.
Then
\be\la{rr}
 0<\rRe (v\lam)<\rRe\sqrt{\lam^2-\mu^2},~~~~~~\lam\in\C\setminus\ov\Gamma
\ee
Denote by $\cL (B_1,B_2)$ the Banach space of bounded linear operators
from a Banach space $B_1$ to a Banach space $B_2$.
Formulas (\re{glam})
implies the following properties of $R_0(\lam)$:
\begin{lemma}\la{sp}(cf. \ci{A, M})\\
i) The operator $R _0(\lam)$ is  analytic function of
$\lam\in\C\setminus\ov\Ga$ with the values in  $\cL(H^0_0,H^1_0)$.
\\
ii) For $\lam\in\Ga$, the convergence (limiting absorption principle)
holds 
\be\la{lap}
R_0(\lam\pm \ve)\to R_0(\lam\pm 0),\quad \ve\to 0+
\ee
in  $\cL (H^0_{\si},H^1_{-\si})$ with $\si>1/2$,
uniformly in $|\lam|\ge |\mu|+r$ for any $r>0$.
\\
iii) The  asymptotics hold
\be\la{g0ee}
R_0(\lam)=B_0^{\pm}\frac{1}{\sqrt\nu}+B_1^{\pm}+{\cal O}(|\nu|^{1/2}),\quad
\nu=\lam\mp\mu\to 0,\quad\lam\in\C\setminus\ov\Gamma
\ee
in   ${\cal L}(H^0_{\si},H^1_{-\si})$ with $\si>5/2$, where
\beqn\la{AB}
B_0^{\pm}&=&{\rm Op}\Big[\frac {e^{\mp\ga^2 v\mu (x-y)}}{2\sqrt{\pm 2\mu}}\Big]
\in{\cal L}(H^0_{\si},H^1_{-\si}),~~\si>1/2\\
\nonumber
B_1^{\pm}&=&{\rm Op}\Big[-\fr{\ga^2e^{\mp\ga^2 v\mu (x-y)}|x-y|}2\Big]
\in{\cal L}(H^0_{\si},H^1_{-\si}),~~\si>3/2
\eeqn 
 \\
iv)
The asymptotics  (\ref{g0ee}) can be differentiated two times:
\be\la{g1ee}
R_0'(\lam)=-B_0^{\pm}\frac{1}{2\nu\sqrt\nu} +{\cal O}\big(|\nu|^{-1/2}\big),\quad
R_0''(\lam)={\cal O}\big(|\nu|^{-5/2}\big),\quad
\nu=\lam\mp\mu\to 0,\quad\lam\in\C\setminus\ov\Gamma
\ee
in  ${\cal L}(H^0_{\si},H^1_{-\si})$ with $\si>5/2$.
\\
v)
For $s\in\R$, $l=-1,0,1,2$, $k=0,1,2,...$ and $\si>1/2+k$ the decay holds
\be\la{A}
  \Vert R_0^{(k)}(\lam)\Vert_{\cL (H^s_\si,H^{s+l}_{-\si})}
  = \cO(|\lam|^{-(1-l)}),\quad |\lam|\to\infty,\quad
  \lam\in\C\setminus\Gamma
\ee
\end{lemma}
\begin{proof}
We prove  the properties ii) and v) since other properties follow
directly from (\re{glam}).\\
{\it Step i)}
First, we prove the convergence (\re{lap}).
The norm of the operator $R_0(\lam): H^0_\si\to H^1_{-\si}$
is equivalent to the norm of the operator
$$
\langle x\rangle^{-\si}R_0(\lam)\langle y\rangle^{-\si}:L^2\to H^1
$$
The norm of the latter operator does not exceed the sum in $k$, $k=0,1$,
of the norms of operators
\begin{equation}\label{pa}
\pa_{x}^{k}
[\langle x\rangle^{-\si}R_0(\lam,x,y)\langle y\rangle^{-\si}]: L^2\to L^2
\end{equation}
According  (\re{glam}) and (\re{rr}),
$$
|\pa_{x}^{k}R_0(\lam, x,y)|\le C(\lam),\quad k=0,1,\quad x,y\in\R,
\quad\lam\in\C\setminus\ov\Gamma
$$
Hence for $\si>1/2$ we have
$$
\sum\limits_k\int|\pa_{x}^{k}
[\langle x\rangle^{-\si}R_0(\lam,x,y)\langle y\rangle^{-\si}]|^2dxdy
\le C(\lam)\int\langle x\rangle^{-2\si}\langle y\rangle^{-2\si}
dxdy \le C_1(\lam)
$$
The estimate  implies that Hilbert-Schmidt norms of operators
(\ref{pa}) is finite.
For $\lam\in\Gamma$ and $x, y\in\R$, 
there exists the pointwise limit 
\begin{equation*} 
R_0(\lam\pm \varepsilon,x,y)\to\;R_0(\lam\pm 0,x,y),\quad\varepsilon\to 0+ 
\end{equation*} 
Therefore,
$$
\sum\limits_k\int
|\pa_{x}^{k}[\langle x\rangle^{-\si}R_0(\lam\pm \ve,x,y)\langle y\rangle^{-\si}-
\langle x\rangle^{-\si}R_0(\lam\pm \ve,x,y)\langle y\rangle^{-\si}]|^2
dxdy\to 0,\quad \ve\to 0+
$$
by the Lebesgue dominated convergence theorem, 
hence \eqref{lap} is proved.\\
{\it Step ii)} Now we prove the decay (\ref{A}).
It suffices to verify the case  $s=0$ since $R_0(\lam)$ commutes with
the operators $\langle \na\rangle^s$ with arbitrary $s\in\R$.
For $k=0$ and $l=0,1,2$ the decay (\re{A}) follows obviously  from (\re{glam}).
In the case  $k=0$ and $l=-1$ the decay follows from the identity
\be\la{idR0}
R_0(\lam)=\frac 1{m^2+\lam^2}
\big(1+\fr{\Delta R_0(\lam)}{\ga^2}+2v\lam\na R_0(\lam)\big)
\ee
Namely, using  (\re{A})  with $l=0$ and $l=1$, we obtain
$$
\Vert \na R_0(\lam)\Vert_{\cL(H^0_\si,H^{-1}_{-\si})}=\cO(|\lam|^{-1}),\quad
\Vert \Delta R_0(\lam)\Vert_{\cL (H^0_\si,H^{-1}_{-\si})}=\cO (1)
$$
hence (\re{idR0}) implies
$$
\Vert R_0(\lam)\Vert_{\cL (H^0_\si,H^{-1}_{-\si})}
=\cO(|\lam|^{-2})
$$
In the case $k\not=0$ the bounds  (\re{A}) follow similarly by 
differentiating (\re{glam}).
\end{proof}
Formula  (\ref{Green}) and Lemma \re{sp} imply
\begin{cor}\la{LAP}
i) The resolvent $\cR _0(\lam)$ is  analytic function of
$\lam\in\C\setminus\ov\Ga$ with the values in  $\cL(\cF_0,\cF_0)$.
\\
ii) For $\lam\in\Gamma$, the convergence (limiting absorption principle) holds
\be\la{lapp}
\cR_0(\lam\pm \ve)\to \cR_0(\lam\pm 0),\quad \ve\to 0+
\ee
in $\cL(\cF_\si,\cF_{-\si})$  with $\si>1/2$.\\
\\
iii) The  asymptotics hold
\be\la{cR-as}
\cR_0(\lam)={\cal B}_0^{\pm}\fr 1{\sqrt\nu}+{\cal B}_1^{\pm}+{\cal O}(|\nu|^{1/2}),
\quad\nu=\lam\mp\mu\to 0,\quad\lam\in\C\setminus\ov\Gamma
\ee
in   ${\cal L}(\cF_{\si},\cF_{-\si})$ with $\si>5/2$, where
\be\la{cAB}
{\cal B}_0^{\pm}=B_0^{\pm}\left(\ba{cc}
\mp i\ga m               & -1 \\
\ga^2 m^2               & \mp i\ga m 
\ea\right)\in{\cal L}(\cF_{\si},\cF_{-\si})~~{\rm with}~~\si>1/2
\ee
and ${\cal B}_1^{\pm}\in{\cal L}(\cF_{\si},\cF_{-\si})$ with $\si>3/2$.
\\
iv)
The asymptotics  (\ref{cR-as}) can be differentiated two times:
\be\la{cR-asd}
\cR_0'(\lam)=-{\cal B}_0^{\pm}\frac 1{2\nu\sqrt\nu}+{\cal O}(|\nu|^{-1/2}),
\quad \cR_0''(\lam)={\cal O}(|\nu|^{-5/2}),\quad
\nu=\lam\mp\mu\to 0,\quad\lam\in\C\setminus\ov\Gamma
\ee
in  ${\cal L}(\cF_{\si},\cF_{-\si})$ with $\si>5/2$.
\\
v) For  $k=0,1,2,...$ and $\si>1/2+k$ the asymptotics hold
  \begin{equation}\label{bR0}
    \Vert \cR _0^{(k)}(\lam) \Vert_{\cL (\cF _{\sigma},\cF _{-\sigma})}
    =\cO(1),\quad |\lam|\to\infty,\quad \lam\in\C\setminus\Ga
  \end{equation}
\end{cor}
Denote by $\cG_v (t)$ the dynamical group of equation (\re{KGE1}).
\begin{cor}
\la{irep}
For $t\in\R$ and $\Psi_0\in\cF_\si$ with $\si>1/2$,
the group $\cG_v (t)$ admits the  integral representation
\be\la{Gint1}
  \cG_v (t)\Psi_0=\frac {1}{2\pi i}\int\limits_\Gamma e^{\lam t}
  \Big[\cR _0(\lam-0)-\cR _0(\lam+0)\Big]\Psi_0~ d\lam
\ee
where the integral converges in the sense of distributions of $t\in\R$
with the values in $\cF_{-\si}$.
\end{cor}
\begin{proof}
Summing up the representations  (\ref{Gint}) and  (\ref{Gints}), and sending $\ve\to 0+$,
we obtain  (\ref{Gint1}) by the Cauchy theorem
and Corollary \re{LAP}.
\end{proof}
\subsection{Time decay}\la{Tdec}
For the integral kernel of the operator $\cG_v(t)$ we have
\be\la{AX}
\cG_v(x-y,t)=\cG_0(x-y-vt,t),~~x,y\in\R,~~t\in\R
\ee
Here
\be\la{J0} 
\cG_0(z,t)=\left( \ba{ll}
\dot G_0(z,t)    &              G_0(z,t)\\
\ddot G_0(z,t)   &        \dot G_0(z,t)
\ea \right),\quad G_0(z,t)=\frac 12\theta(t-|z|)J_0(m\sqrt{t^2-z^2}),
\quad z=x-y
\ee
where $J_0$ is the Bessel function. 
The relation (\re{AX}) implies the Huygen's principle for the group $\cG_v(t)$, i.e.
$$
\cG_v(x-y,t)=0,~~~~~|x-y-vt|>t
$$
Also, the relation (\re{AX}) implies the energy conservation for the group $\cG_v(t)$.
Namely, for $\Psi(t)=(\psi(\cdot,t),\pi(\cdot,t))=\cG_v(t)\Psi_0$ we have
$$
\int[|\pi(x,t)+v\cdot\na\psi(x,t)|^2+|\na\psi(x,t)|^2+m^2|\psi(x,t)|^2]dx=\const,
~~~~~~t\in\R
$$
In particular, this gives  that
$$
\Vert \Psi(t)\Vert_{\cF_0}\le C\Vert\Psi_0\Vert_{\cF_0},\,\,\,t\in\R
$$
We represent $\cG_v(z,t)$ as
$$
\cG_v(z,t)=\cG_{b}(z,t)+{\cal G}_{r}(z,t),\quad z\in\R,\quad t\ge 0
$$
where
\be\la{G0} 
\cG_{b}(z,t):=\frac {1}{\sqrt{2m\pi
t/\ga}}\left( \ba{cc} -\fr m{\ga}\sin[m(\fr{t}{\ga}+\ga vz)-\frac{\pi}4]  &
\cos[m(\fr{t}{\ga}+\ga vz)-\frac{\pi}4]
 \\\\
-\fr{m^2}{\ga^2}\cos[m(\fr{t}{\ga}+\ga vz)-\frac{\pi}4]  &
-\fr m{\ga}\sin[m(\fr{t}{\ga}+\ga vz)-\frac{\pi}4]
  \ea \right)
\ee
The entries  of the matrix $\cG_{b}(z,t)$ admit the bounds
\be\la{G0-est}
|\cG_{b}^{ij}(z,t)|\le C(v)/\sqrt t,\quad i,j=1,2, \quad z\in\R,\quad t\ge 1
\ee
The  group  ${\cal G}_v(t)$ slow decays, like $t^{-1/2}$.
We will show that $\cG_{b}(t)={\rm Op}[\cG_b(x-y,t)]$ is only term 
responsible for the slow decay.
More exactly,  in the next section we will prove the following basic
proposition
\begin{pro}\label{p2}
The decay holds
 \begin{equation}\label{Gs2}
  \cG_r(t)={\rm Op}[\cG_r(x-y,t)]={\cal O}(t^{-3/2}),\quad t\to\infty
  \end{equation}
in the norm of $\cL (\cF _{\sigma},\cF _{-\sigma})$  with $\sigma>5/2$.
\end{pro}

The following key observation is  that (\ref{G0}) contains just two frequencies
$\pm\mu$ which are the edge  points of the continuous spectrum.
This suggests that the  term $\cG_b(t)$ with ``bad decay'' $t^{-1/2}$
should not contribute to the high energy component of the 
group  ${\cal G}_v(t)$ and the high energy component of the group ${\cal G}_v(t)$
decays like $t^{-3/2}$.

More precisely, let us introduce the  {\it low energy} and
{\it high energy} components of $\cG_v (t)$:
\begin{equation}\label{Gs}
{\cal G}_l(t)=\frac 1{2\pi i}\int\limits_\Gamma e^{\lam t}l(i\lam)
  \Big[\cR _0(\lam-0)-\cR _0(\lam+0)\Big]~ d\lam
\end{equation}
\begin{equation}\label{Gh}
{\cal G}_h(t)=\frac 1{2\pi i}\int\limits_\Gamma e^{\lam t}h(i\lam)
  \Big[\cR _0(\lam-0)-\cR _0(\lam+0)\Big]~ d\lam
\end{equation}
where
$l(\om)\in C_0^\infty(\R)$ is an even function,
$l(\om)=0$  if $|\om|>|\mu|+2\ve$, and
$l(\om)=1$ if $|\om|\le |\mu|+\ve$ with an $\ve>0$,
and $h(\om)=1-l(\om)$.
\begin{theorem}\la{TD}
In  $\cL (\cF _{\sigma},\cF _{-\sigma})$ with $\sigma>5/2$  the decay holds
\begin{equation}\label{Gb1}
{\cal G}_h(t)={\cal O}(t^{-3/2}),\quad t\to\infty
\end{equation}
\end{theorem}
\begin{proof}
We deduce asymptotics (\ref{Gb1}) from  Proposition \ref{p2}.\\
{\it Step i)}
Let $\Psi_0\in\cF_\si$.
Denote
$$
\Psi^+(t)=\theta(t)\cG_v(t)\Psi_0,\;\Psi_b^+(t)=\theta(t)\cG_b(t)\Psi_0,\;
\Psi_h^+(t)=\theta(t)\cG_h(t)\Psi_0,\;\Psi_r^+(t)=\theta(t)\cG_r(t)\Psi_0
$$
Then
\beqn\nonumber
\Psi_h^+(t)&=&-\frac 1{2\pi }\int\limits_\R e^{i\om t}h(\om)
\cR _0(i\om+0)\Psi_0d\om\\
\nonumber
&=&\fr 1{2\pi }\int\limits_\R e^{i\om t}h(\om)\tilde\Psi^+(i\om)d\om
=\fr 1{2\pi }\int\limits_\R e^{i\om t}h(\om)
\Big[\tilde\Psi_b^+(i\om)+\tilde\Psi_r^+(i\om)\Big]~d\om\\
\la{pl}
&=&\Psi_r^+(t)+\fr 1{2\pi }\int\limits_\R e^{i\om t}h(\om)
\tilde\Psi_b^+(i\om)d\om-\fr 1{2\pi }\int\limits_\R e^{i\om t}l(\om)
\tilde\Psi_r^+(i\om)d\om
\eeqn
where $\tilde\Psi^+(\lam)=\int\limits_0^{\infty}e^{-\lam t}\Psi^+(t)dt$ 
and so on. By (\ref{Gs2})
\be\la{1t}
\Vert\Psi_r^+(t)\Vert_{\cF_{-\si}}=\cO(t^{-3/2}),\quad t\to\infty
\ee
{\it Step ii)}
Let us consider the second summand in the last line of  (\ref{pl}).
By (\ref{G0}) the vector function $\tilde\Psi_b^+(i\om)$ is
a smooth function for $|\om|>|\mu|+\ve$, and
$\pa_\om^k\tilde\Psi^+_b(i\om)=\cO(|\om|^{-1/2-k})$, $k=0,1,2...$, $\om\to\infty$.
Hence partial integration implies that
\be\la{st}
\Big\Vert\int\limits_\R e^{i\om t}h(\om)
\tilde\Psi_b^+(i\om)d\om\Big\Vert_{\cF_{-\si}}=\cO(t^{-N}),
\quad\forall N\in\N,\quad t\to\infty
\ee
{\it Step iii)}
Finally, let us consider the third summand in the last line of
(\ref{pl}).
Introducing the function $L(t)$ such that $\tilde L(\lam)=l(i\lam)$, we
obtain
\be\la{st1}
\fr 1{2\pi }\int\limits_\R e^{i\om t}l(\om)
\tilde\Psi_r^+(i\om)d\om= [L\star\Psi_r^+](t)=\cO(t^{-3/2}),\quad t\to\infty
\ee
in the norm of $\cF_{-\si}$,
since  $L(t)=\cO(t^{-N})$, $t\to\infty$  for any $N\in\N$, and
$\Vert\Psi_r^+(t)\Vert_{\cF_{-\si}}=\cO(t^{-3/2})$ by (\ref{Gs2}).
Finally, (\ref{pl})- (\ref{st1}) imply (\ref{Gb1}).
\end{proof}
\subsection{Proof of Proposition \ref{p2}}
Let us  fix an arbitrary  $\ve\in(|v|,1)$.  Denote $\ve_1=\ve-|v|$. 
For any $t\ge 1$ we split the initial function $\Psi_{0}\in{\cF _\si}$
in two terms,
$\Psi_{0}=\Psi_{0,t}^{\prime }+\Psi_{0,t}^{\prime \prime }$,
$\Psi'_{0,t}=(\psi'_{0,t}, \pi'_{0,t})$, $\Psi''_{0,t}=(\psi''_{0,t}, \pi''_{0,t})$,
such that
\begin{equation} \label{FFn}
  \Vert \Psi_{0,t}^{\prime}\Vert _{\cF _\si}+\Vert \Psi_{0,t}^{\prime\prime }
  \Vert _{\cF _\si}\le C\Vert \Psi_{0}\Vert _{\cF _\si},\quad t\ge 1
\end{equation}
\be\label{F}
  \Psi_{0,t}^{\prime }(x)=0 \quad{\rm for}\quad |x|>\frac{\ve_1 t}{2},
\quad {\rm and}\quad
  \Psi_{0,t}^{\prime \prime }(x)=0 \quad{\rm for}\quad |x|<\frac{\ve_1 t}{4}
\ee
We estimate $\cG_r(t)\Psi'_{0,t}$ and $\cG_r(t)\Psi''_{0,t}$ separately.
\medskip\\
{\it Step i)}
First we consider
$\cG_r(t)\Psi''_{0,t}=\cG_v(t)\Psi''_{0,t}-\cG_b(t)\Psi''_{0,t}$.
Using energy conservation  and properties (\ref {FFn})- (\ref{F}) we obtain
\be\label{zxc}
  \Vert \cG_v (t)\Psi_{0,t}^{\prime\prime }\Vert _{\cF _{-\si}} \le
  \Vert \cG_v(t)\Psi_{0,t}^{\prime\prime }\Vert _{\cF_0}
  \le C\Vert \Psi_{0,t}^{\prime \prime }\Vert _{\cF_0}
 \leq C(\ve)t^{-\si}\Vert \Psi_{0,t}^{\prime \prime }\Vert _{\cF _\si}
  \leq C_1(\ve)t^{-5/2}\Vert \Psi_{0}\Vert _{\cF _\si},\quad t\geq 1
\ee
since $\si>5/2$.
Further,  (\ref{G0-est}) and the Cauchy inequality imply
\beqn\nonumber
|(\cG_b^{22}(t)\pi''_{0,t})(y)|&\le&\ds\frac{C}{\sqrt t}\Big|\int\pi''_{0,t}(x)dx\Big|
\le\frac{C}{\sqrt t}\Big(\int|\pi''_{0,t}(x)|^2(1+x^2)^{\si}
dx\Big)^{1/2}\Big(\int\limits_{\ve_1 t/4}^\infty\frac{dx}{(1+x^2)^{\si}}\Big)^{1/2}\\
\la{Ce}
&\le&\ds\frac{C(\ve)}{\sqrt t}t^{-\si+1/2}\Vert\pi''_{0,t}\Vert_{H^0_\si}
\le C(\ve)t^{-5/2}\Vert\pi''_{0,t}\Vert_{H^0_\si},\quad t\ge 1
\eeqn
Hence
$\Vert\cG_b^{22}(t)\pi''_{0,t}\Vert_{H^0_{-\si}}\le
C(\ve)t^{-5/2}\Vert\pi''_{0,t}\Vert_{H^0_\si}$.
The functions $\cG_b^{12}(t)\pi''_{0,t}$ and
$\cG_b^{i1}(t)\psi''_{0,t}$, $i=1,2$ can be estimated
similarly. Therefore,
\be\la{zc}
 \Vert \cG_b(t)\Psi''_{0,t}\Vert _{\cF _{-\si}} \le
 C(\ve)t^{-5/2}\Vert \Psi_{0}\Vert _{\cF _\si},\quad t\geq 1
\ee
and (\ref{zxc})- (\ref{zc}) imply that
\be\la{Gr1}
\Vert \cG_r(t)\Psi''_{0,t}\Vert _{\cF _{-\si}} \le
 C(\ve)t^{-5/2}\Vert \Psi_{0}\Vert _{\cF _\si},\quad t\geq 1
\ee
{\it Step ii)}
Denote by $\zeta $  the operator of multiplication by the function $\zeta ({|x|}/{t})$,
where $\zeta =\zeta (s)\in C_{0}^{\infty}(\R)$,
$\zeta (s)=1$ for $|s|<\ve_1/4$, $\zeta (s)=0$ for $ |s|>\ve_1/2$.
Obviously, for any $k$, we have
$$
 |\partial _{x}^{k}\zeta ({|x|}/{t})|\leq C(\ve)<\infty,~~~~~~t\ge 1
$$
Since  $1-\zeta ({|x|}/{t})=0$ for $|x|<\ve_1 t/4$, then by the energy
conservation and (\ref {FFn}), we obtain
\be\label{zaz}
  ||(1-\zeta )\cG_v (t)\Psi_{0,t}^{\prime }||_{\cF _{-\si}}
  \le C(\ve)t^{-\si}||\cG_v (t)\Psi_{0,t}^{\prime }||_{\cF_0}
  \le C_1(\ve)t^{-\si}||\Psi_{0,t}^{\prime }||_{\cF_0}
  \leq C_2(\ve)t^{-5/2}||\Psi_{0}||_{\cF_\si} ,~~t\ge 1
\ee
Further,  (\ref{G0-est})  and the Cauchy inequality imply,
similarly (\re{Ce}), that
$$
|(\cG_b^{22}(t)\pi'_{0,t})(y)|\le\ds\frac{C}{\sqrt t}\Big|\int\pi'_{0,t}(x)dx\Big|
\le\frac{C}{\sqrt t}\Vert\pi'_{0,t}\Vert_{H^0_\si}
$$
Hence, we obtain
$$
\Vert(1-\zeta )\cG_b^{22}(t)\pi'_{0,t}\Vert_{H^0_{-\si}}
\le\frac{C}{\sqrt t}
\Vert\pi'_{0,t}\Vert_{H^0_\si}\Big(\int\limits_{\ve_1 t/4}^\infty
\frac{dy}{(1+y^2)^{\si}}\Big)^{1/2}\le
C(\ve)t^{-5/2}\Vert\pi'_{0,t}\Vert_{H^0_\si}
$$
The functions $(1-\zeta )\cG_b^{12}(t)\pi'_{0,t}$ and
$(1-\zeta )\cG_b^{i1}(t)\psi'_{0,t}$, $i=1,2$ can be estimated
similarly. Hence,
\be\la{zz}
 ||(1-\zeta )\cG_b(t)\Psi_{0,t}^{\prime }||_{\cF _{-\si}}
\leq C(\ve)t^{-5/2}||\Psi_{0}||_{\cF_\si} ,~~~~~t\ge 1
\ee
and  (\ref{zaz}) - (\ref{zz}) imply
\be\la{zrz}
 ||(1-\zeta )\cG_r (t)\Psi_{0,t}^{\prime }||_{\cF _{-\si}}
\leq C(\ve)t^{-5/2}||\Psi_{0}||_{\cF_\si} ,~~~~~t\ge 1
\ee
\\
\textit{Step iii)}
Finally, let us estimate $\zeta \cG_r(t)\Psi_{0,t}^{\prime }$.
Let $\chi _{ t}$ be the characteristic function of the ball $ |x|\le \ve_1 t/2$.
We will use the same notation for the operator of multiplication
 by this characteristic function. By (\ref{F}), we have
\be\la{xx}
  \zeta \cG _r(t)\Psi_{0,t}^{\prime}=\zeta \cG _r(t)\chi_{ t}\Psi_{0,t}^{\prime}
\ee
The matrix kernel  of the operator $\zeta \cG _r(t)\chi _{ t}$ is equal to
$$
  \cG _r^{\prime }(x-y,t)=\zeta ({|x|}/{t})\cG _r(x-y,t)\chi _{ t}(y)
$$
Well known asymptotics of the Bessel function \cite{W} imply the following lemma, which we
prove in Appendix.
\begin{lemma}\la{l1}
For any $\ve\in (|v|,1)$ the bounds hold
\be\la{G2R}
|{\partial^{k}_z\cal G}_r(z,t)|\le C(\ve )(1+z^2)t^{-3/2},
\quad|z|\le(\ve-|v|) t, \quad t\ge 1, \quad k=0,1
\ee
\end{lemma}
Since $\zeta ({|x|}/{t})=0$ for $|x|>\ve_1 t/2 $ and $\chi _{ t}(y)=0$
for $|y|>\ve_1 t/2$ then $ \cG _r^{\prime }(x-y,t)=0$ for $|x-y|>\ve_1t=(\ve-|v|)t$.
 Hence, (\ref{G2R})\ imply that
\begin{equation}\label{qaz}
  |\pa _{x}^{k}\cG_r^{\prime}(x-y,t)|\le C(\ve)(1+(x-y)^2)t^{-3/2},
  \quad k=0,1,\quad t\ge 1
\end{equation}
The norm of the operator
$\zeta \cG_r(t)\chi_{t}: \cF _{\si}\rightarrow \cF _{-\si}$
is equivalent to the norm of the operator
$$
  \langle x\rangle^{-\si}\zeta \cG _r(t)\chi _{t}(y)\langle y\rangle^{-\si}:
  \cF_0 \rightarrow \cF_0
$$
The norm of the later operator does not exceed the sum in $k$, $k=0,1$
of the norms of operators
\begin{equation}\label{1234}
  \pa_{x}^{k}[\langle x\rangle^{-\si}\zeta \cG_r(t)\chi_{ t}(y)
  \langle y\rangle^{-\si}]: L^2(\R)\oplus L^2(\R)
\to
  L^2(\R)\oplus L^{2}(\R)
\end{equation}
The bounds (\ref{qaz}) imply that the Hilbert-Schmidt norms of 
operators (\ref{1234}) do not exceed $C(\ve) t^{-3/2}$ since $\si>5/2$. 
Hence, (\ref{FFn}) and (\ref{xx}) imply that
\begin{equation} \label{HS}
  ||\zeta \cG _r(t)\Psi_{0,t}^{\prime }||_{\cF _{-\si}}\le C(\ve) t^{-3/2}|
  |\Psi_{0,t}^{\prime }||_{\cF_\si}
  \le C_1(\ve) t^{-3/2}||\Psi_{0}||_{\cF _\si}, \quad t\ge 1
\end{equation}
Finally,  (\ref{zrz}) and  (\ref {HS}) imply
$$
 ||\cG _r(t)\Psi_{0,t}^{\prime }||_{\cF _{-\si}}
  \le C(\ve) t^{-3/2}||\Psi_{0}||_{\cF _\si},\quad t\ge 1
$$
Proposition  \ref{p2} is proved.
\setcounter{equation}{0}
\section{Perturbed  equation}\la{pKG}
\subsection{Perturbed resolvent}
Now we consider the resolvent of the perturbed equation. We  use the formula
\be\la{fR}
\cR(\lam)=(1+\cR_0(\lam){\cV})^{-1}\cR_0(\lam),\quad{\cV}=\left(\ba{cc}
 0  & 0 \\
-V  & 0
\ea\right)
\ee
By (\re{Green}) we have
\be\la{ffR}
(1+\cR_0(\lam){\cV})^{-1}=\left(\ba{cc}
1+ R_0(\lam)V                      &   0 \\\\
-(v\nabla\!-\!\lam) R_0(\lam)V        &   1
\ea\right)^{-1}\!\!
=\left(\ba{cc}
(1+R_0(\lam)V)^{-1}                                   &   0 \\\\
(v\nabla\!-\!\lam)\Big(1\!-\!(1+R_0(\lam)V)^{-1}\Big)         &   1
\ea\right)
\ee
Let us denote 
$$
\cH=-(1-v^2)\Delta+m^2+V,\quad~~
R(\lam)=(\cH+\lam^2-2v\lam\nabla)^{-1}=(1+R_0(\lam)V)^{-1}R_0(\lam)
$$
Substituting (\ref{ffR}) into (\ref{fR}) we obtain
$$
\cR(\lam)=\left(\ba{cc}
(1+R_0(\lam)V)^{-1}                                   &   0 \\\\
(v\nabla-\lam)\Big(1-(1+R_0(\lam)V)^{-1}\Big)         &   1
\ea\right)\left(\ba{cc}
(v\nabla-\lam)R_0(\lam) &   -R_0(\lam) \\\\
(-\Delta+m^2)R_0(\lam)  & (v\nabla-\lam)R_0(\lam)
\ea\right)
$$
\be\la{Rf}
=\left(\ba{cc}
 R(\lam)(v\nabla-\lam)                &        -R(\lam) \\\\
1-(v\nabla-\lam)R(\lam)(v\nabla-\lam) & (v\nabla-\lam)R(\lam)
\ea\right)
\ee
Similarly (\re{cH0})-(\re{RtilR}), we obtain
\be\la{cH}
(\cH+\lam^2-2v\lam\na)\psi(x)=e^{-\ga^2v\lam x}(\cH+\ga^2\lam^2)
e^{\ga^2v\lam x}\psi(x)
\ee
\be\la{RtilR1}
R(\lam)=e^{-\ga^2v\lam x}\ga^2\tilde R(\ga^2m^2+\ga^4\lam^2)e^{\ga^2v\lam y}
\ee
where
$\tilde R(\zeta)=(-\Delta+\zeta+V\ga^2)^{-1}$ is the resolvent
of the Schr\"odinger operator $-\Delta+V\ga^2$.
\subsection{Spectral properties}
To prove the long time decay for the perturbed  equation,
we first establish the spectral properties of the generator.
\subsubsection{Limiting absorption principle}
\begin{pro}\la{sp1}
Let the potential $V$ satisfy  (\re{V}). Then\\
i)$R(\lam)$ is  meromorphic function of
$\lam\in\C\setminus\ov\Gamma$ with the values in $\cL(H^0_0,H^1_0)$;\\
ii) For $\lam\in\Ga$, the convergence holds 
\be\la{lap1}
R(\lam\pm \ve)\to R(\lam\pm 0),\quad\ve\to 0+
\ee
in  $\cL (H^0_{\si},H^1_{-\si})$ with $\si>1/2$,
uniformly in $|\lam|\ge |\mu|+r$ for any $r>0$.
\end{pro}
\begin{proof}
{\it Step i)} The statement i) follows from Lemma \re{sp}-i), the Born
splitting
\be\la{Born}
R(\lam)=R_0(\lam)(1+VR_0(\lam))^{-1}
\ee
and the  Gohberg-Bleher theorem \cite{B, GK}
since  $VR_0(\lam)$
is a compact operator in $L^2$ for $\lam\in\C\setminus\ov\Gamma$.\\
{\it Step ii)}  
The convergence  (\re{lap1})  follow from 
(\re{lap}) by the Born splitting (\re{Born}) if
$$
[1+VR_0(\lam\pm \ve)]^{-1}\to[1+VR_0(\lam\pm 0)]^{-1},
\quad\ve\to +0,\quad\lam\in\Gamma
$$
in $\cL(H^0_{\si};H^0_{\si})$.
This convergence holds if and only if both limit operators
$1+VR_0(\lam \pm 0)$ are invertible in $H^0_{\si}$ for $\lam\in\Gamma$.
The operators are invertible according to the reversibility of the
operators $1+\ga^2V\tilde R_0(\zeta \pm i0)$ in $H^0_{\si}$ for
$\zeta<0$ (see \cite[Theorem 3.3 and Lemma 4.2] {A}) and the relations
$$
1+VR_0(\lam \pm 0)=e^{-\ga^2v\lam x}
\big(1+\ga^2V\tilde R_0(\ga^2m^2+\ga^4(\lam\pm i0)^2)\big)e^{\ga^2v\lam y}
$$
which follows from (\re{RtilR}).
\end{proof}
Formula (\re{Rf}) and Proposition \re{sp1} imply
\begin{cor}\la{LP}
Let  the  conditions (\ref{V}) holds. Then
\\
i) $\cR(\lam)$ is  meromorphic function of
$\lam\in\C\setminus\ov\Ga$ with the values in $\cL(\cF_0,\cF_0)$;
\\
ii) For $\lam\in\Gamma$, the convergence  holds
\be\la{LPP}
\cR(\lam\pm \ve)\to \cR(\lam\pm 0),\quad \ve\to 0+
\ee
in $\cL(\cF_\si,\cF_{-\si})$  with $\si>1/2$.
\end{cor}
\subsubsection{High energy decay}
\begin{lemma}\la{hd}
For  $k=0,1,2$, $s=0,1$ and  $l=-1,0,1$ with $s+l\in \{0,1\}$ the asymptotics hold
\be\la{AR}
\Vert R^{(k)}(\lam\pm 0)\Vert_{\cL (H^s_\si,H^{s+l}_{-\si})}
={\cal O}(|\lam|^{-(1-l+k)}),\quad |\lam|\to\infty,\quad
\lam\in\Gamma
\ee
with $\si>1/2+k$.
\end{lemma}
\begin{proof}
The decay  follows from formula (\re{RtilR1}) and the known decay of 
Schr\"odinger resolvent $\tilde R(\zeta)$ (see \cite{A, jeka, M, 1dkg}). 
\end{proof}
\begin{cor}\la{chd}
For $k=0,1,2$ and $\si>1/2+k$ the  asymptotics hold
\be\la{bR}
 \Vert\cR^{(k)}(\lam\pm 0)\Vert_{\cL (\cF_\si,\cF_{-\si})}
 ={\cal O}(1),\quad |\lam|\to\infty,\quad\lam\in\Gamma
\ee
\end{cor}
The resolvents $\cR(\lam)$ and $\cR_0(\lam)$ are related by the Born
perturbation series
\be\la{id}
  \cR (\lam)
  = \cR _0(\lam)-\cR _0(\lam)\cV \cR _0(\lam)
  +\cR _0(\lam)\cV \cR _0(\lam)\cV \cR (\lam),
  \quad\lam\in\C\setminus[\Ga\cup\Si]
\ee
where $\Si$ is the set of eigenvalues of the
operator $\cA$.
An important role in (\re{id}) plays the product
${\cal W}(\lam):= \cV\cR _0(\lam)\cV$. Now we obtain the
asymptotics of ${\cal W}(\lam)$ for large $\lam$.
\begin{lemma}\la{large1}
Let  $k=0,1,2$, and the potential $V$ satisfy (\ref{V}) with $\beta>1/2+k+\si$
where $\si>0$. Then  the asymptotics hold
\begin{equation}\label{bRV}
  \Vert{\cal W}^{(k)}(\lam)\Vert_{\cL (\cF _{-\sigma},\cF _{\sigma})}=
  {\cal O}(|\lam|^{-2}),\quad |\lam|\to\infty,\quad \lam\in\C\setminus\ov\Ga
  \end{equation}
\end{lemma}
\begin{proof}
Asymptotics (\ref{bRV})  follow from  the algebraic structure of the matrix
$$
  {\cal W}^{(k)}(\lam)=\cV\cR _0^{(k)}(\lam)\cV =\left(\begin{array}{cc}
  0                                      &   0
  \\
  -VR_0^{(k)}(\lam)V            &   0
\end{array}\right)
$$
since (\ref{A}) with $s=1$ and $l=-1$ implies that
$$
 \Vert VR _0^{(k)}(\lam)V f\Vert_{H^{0}_{\si}}\le C\Vert R_0^{(k)}
 (\lam)V f\Vert_{H^{0}_{\si-\beta}}
 ={\cal O}(|\lam|^{-2})\Vert V f\Vert_{H^{1}_{\beta-\si}}
 ={\cal O}(|\lam|^{-2})\Vert f\Vert_{H^{1}_{-\si}}
$$
since $\beta-\si>1/2+k$.
\end{proof}
\subsubsection{Low energy expansions}
\begin{pro}\la{AP}
The asymptotics hold
\be\la{RRR}
\left.\ba{lll}
\cR(\lam)={\cal B}^{\pm}+{\cal O}(\nu^{1/2})\\
\cR'(\lam)={\cal O}(\nu^{-1/2})\\
\cR''(\lam)={\cal O}(\nu^{-3/2})\ea\right|
\quad\nu:=\lam\mp\mu\to 0,\quad\lam\in\C\setminus\Gamma
\ee
in the norm ${\cal L}(\cF_{\si},\cF_{-\si})$ with $\si>5/2$,
where ${\cal B}^{\pm}\in{\cal L}(\cF_{\si},\cF_{-\si})$ does not depend on $\lam$.
\end{pro}
First we prove the boundedness of the resolvent near the points $\pm\mu$.
\begin{lemma}\la{bnd}
Let the conditions  (\re{V}) and (\ref{SC}) hold. Then
the families $\{\cR(\pm\mu+\ve): \pm\mu+\ve\in\C\setminus\ov\Gamma, |\ve|<\de\}$
are bounded in the operator norm
of ${\cal L}(\cF_{\si},\cF_{-\si})$ for any $\si>3/2$ and
sufficiently small $\de$.
\end{lemma}
\begin{proof}
Let us consider the equation for eigenfunctions of operator $\cA$
with eigenvalues $\lam=\pm\mu$:
$$
\left(
\ba{cc}
v\nabla          & 1 \\
\Delta-m^2- V    & v\nabla
\ea
\right)\left(
\ba{c}
\psi \\ \pi
\ea
\right)=\pm\mu\left(\ba{c}
\psi \\ \pi
\ea\right),\quad \Psi=\left(\ba{c}\psi \\ \pi\ea\right)\in\cF_0
$$
From the first equation we have $\pi=-(v\nabla\mp\mu)\psi$.
Then the second equation becomes
\be\la{mu}
(\cH+\mu^2\mp 2v\mu\nabla)\psi=
e^{\mp i\ga vmx}(-\fr 1{\ga^2}\Delta+V)e^{\pm i\ga vmx}\psi=0
\ee
Hence, the condition (\ref{SC}) implies that $\Psi=0$.
Similarly,  (\ref{SC}) implies that the equation 
$\cA\Psi=\pm\mu\Psi$ has no nonzero solutions $\Psi\in\cF_{-1/2-0}$.
Then the required boundedness of the resolvent near the points $\pm\mu$
follows similarly to  \cite [Theorem 7.2 ]{M}.    
\end{proof}
This lemma implies that
the operators $(1+\cR_0(\lam)\cV)^{-1}=1-\cR(\lam)\cV$
and $(1+\cV\cR_0(\lam))^{-1}=1-\cV\cR(\lam)$  are bounded
in ${\cal L}(\cF_{-\si},\cF_{-\si})$ and in
${\cal L}(\cF_{\si},\cF_{\si})$ respectively
for $|\lam\mp\mu|<\de$, $\lam\in\C\setminus\overline\Gamma$.
Now we  prove more detailed asymptotics 
\begin{lemma}\la{A0}
The asymptotics hold
\be\la{A00}
(1\!+\!\cR_0(\lam)\cV)^{-1}{\cal B}_0^{\pm}\!={\cal O}(\sqrt\nu),\quad
{\cal B}_0^{\pm}(1+\cV\cR_0(\lam))^{-1}\!={\cal O}(\sqrt\nu),\quad\nu=\lam\mp\mu\to 0,
\quad\lam\in\C\setminus\Gamma
\ee
in ${\cal L}(\cF_{\si},\cF_{-\si})$ with $\si>3/2$. 
\end{lemma}
\begin{proof}
The asymptotics (\ref{cR-as}) implies
$$
\left.\ba{ll} \cR(\lam)=\big(1+\cR_0(\lam)\cV\big)^{-1}\cR_0(\lam)
=\big(1+\cR_0(\lam)\cV\big)^{-1}\big({\cal B}_0^{\pm}\ds\frac{1}{\sqrt\nu}
+{\cal O}(1)\big)\\\\
\cR(\lam)=\cR_0(\lam)\big(1+\cV\cR_0(\lam)\big)^{-1}
=\big({\cal B}_0^{\pm}\ds\frac{1}{\sqrt\nu}+{\cal O}(1)\big)
\big(1+\cV\cR_0(\lam)\big)^{-1}\ea\right|
~\nu=\lam\mp\mu\to 0, ~~\lam\in\C\setminus\Gamma
$$
Hence, the boundedness $\cR(\lam)$, $(1+\cR_0(\lam)\cV)^{-1}$ 
and $(1+\cV\cR_0(\lam))^{-1}$ at the points $\lam=\pm\mu$
in corresponding norms imply the asymptotics (\ref{A00}).
\end{proof}
\begin{cor}\la{exp}
i) The asymptotics hold
\be\la{G00}
\Vert(1+\cR_0(\lam)\cV)^{-1}[e^{\mp i\ga vm x}]\Vert_{\cF_{-\si}}={\cal O}(\sqrt\nu),
\quad\nu=\lam\mp\mu\to 0, \quad\lam\in\C\setminus\Gamma,\quad\si>3/2
\ee
ii) For any $f\in \cF_{\si}$ with $\si>3/2$
\be\la{G01}
\int e^{\pm i\ga vmx}[(1+\cV\cR_0(\lam))^{-1}f](x)dx={\cal O}(\sqrt\nu),
\quad\nu=\lam\mp\mu\to 0, \quad\lam\in\C\setminus\Gamma
\ee
\end{cor}
{\bf Proof of Proposition  \re{RRR}}.
Taking into account the identities
$$
\cR'=(1+\cR_0\cV)^{-1}\cR'_0(1+\cV\cR_0)^{-1},\quad
\cR''=\Big[(1+\cR_0\cV)^{-1}\cR''_0-2\cR'\cV\cR'_0\Big](1+\cV\cR_0)^{-1}
$$
we obtain from (\ref{cR-asd}) and (\ref{G00})-(\ref{G01}) 
the asymptotics (\ref{RRR}) for the derivatives. The asymptotics (\ref{RRR})
for $\cR(\lam)$ follows by integration of asymptotics for $\cR'(\lam)$. 
Proposition \re{RRR} is proved.
\begin{cor}
Let the conditions  (\re{V}) and (\ref{SC}) hold. Then the set $\Si$
of eigenvalues of the operator $\cA$ is finite, i.e.
$\Si=\{\lam_j,~~j=1,...,N\}$.
\end{cor}
\subsection{Time decay}
 Our main result is 
\begin{theorem}\la{main}
Let conditions (\ref{V}) and (\ref{SC}) hold. Then
\begin{equation}\la{full1}
   \Vert e^{t{\cal A}}-\sum\limits_{\om_j\in\Sigma}
   e^{\lam_j t}P_j\Vert_{\cL (\cF _\si,\cF _{-\si})}
  ={\cal O}(|t|^{-3/2}),\quad t\to\pm\infty
\end{equation}
with  $\sigma>5/2$, where  $P_j$ are the Riesz
projections onto the corresponding eigenspaces.
\end{theorem}
\begin{proof}
Corollaries \re{LP} and \re{chd} and Proposition \re{AP}
imply similarly to (\re{Gint1}), that
$$
  \Psi(t)-\sum\limits_{\lam_j\in\Sigma}e^{\lam_j t}P_j\Psi_0=
  \fr 1{2\pi i}\int\limits_\Gamma e^{\lam t}
  \Big[\cR (\lam-0)-\cR (\lam+0)\Big]\Psi_0~ d\lam
=\Psi_l(t)+\Psi_h(t)
$$
where $P_j\Psi_0:=\ds\fr 1{2\pi i}\int_{|\lam-\lam_j|=\delta}\cR(\lam)\Psi_0 d\lam$
with a small $\delta>0$, and low and high energy components are
defined by
\be\la{idl}
\Psi_l(t)=\fr 1{2\pi i}\int\limits_\Gamma l(i\lam)e^{\lam t}
\Big[\cR (\lam-0)-\cR (\lam+0)\Big]\Psi_0~ d\lam
\ee
\be\la{idh}
\Psi_h(t)=\fr 1{2\pi i}\int\limits_\Gamma h(i\lam)e^{\lam t}
\Big[\cR (\lam-0)-\cR (\lam+0)\Big]\Psi_0~ d\lam
\ee
where $l(i\lam)$ and $h(i\lam)$ are defined in Section \ref{Tdec}.
We analyze $\Psi_l(t)$ and $\Psi_h(t)$ separately.
\subsubsection{Low energy component}
We prove the desired decay of $\Psi_l(t)$ using
a special case of Lemma 10.2 from \cite{jeka}.
We consider only the integral (\re{idl}) over $(\mu,\mu+2i\ve)$.
The integral over $(-\mu-2i\ve,-\mu)$ is dealt with in the same way.
Denote by ${\bf B}$  a Banach space with the norm $\Vert\cdot\Vert\,.$
\begin{lemma}\label{jk}
Let $F\in C([a, b],\, {\bf B})$, satisfy
$$
  F(a)=F(b)=0,~~~~
  \Vert F''(\om)\Vert=\cO (|\om-a|^{-3/2}),~~~\om\to a
$$
Then
$$
  \int\limits_a^b e^{-it\omega}F(\omega)d\omega =\cO (t^{-3/2}),
  \quad t\to\infty
$$
\end{lemma}
Due  to (\ref{RRR}), we can apply Lemma \ref{jk} with
$\om=-i\lam$,
$F=l(\om)\big(\cR (i\om-0)-\cR (i\om+0)\big)$,
${\bf B}= \cL (\cF _{\si},\cF _{-\si})$, $a=|\mu|$,
$b=|\mu|+2\ve$ and $\si>5/2$, to get
$$
  \Vert \Psi_l(t)\Vert_{\cF _{-\si}}
 \le C(1+|t|)^{-3/2}\Vert \Psi_0\Vert_{\cF _\si},
  \quad t\in\R,\quad \sigma>5/2
$$
\subsubsection{High energy component}
Let us substitute the series (\re{id}) into the spectral representation
(\re{idh}) for $\Psi_h(t)$:
\beqn\nonumber
  \Psi_h(t)
  &=&\fr 1{2\pi i}\int\limits_\Gamma e^{\lam t}h(i\lam)
  \Big[\cR _0(\lam-0)-\cR _0(\lam+0)\Big]\Psi_0~ d\lam\\
  \nonumber
  &+&\fr 1{2\pi i}\int\limits_\Gamma e^{\lam t}h(i\lam)
\Big[\cR _0(\lam-0)
  \cV \cR _0(\lam-0)-\cR _0(\lam+0)\cV \cR _0(\lam+0)\Big]\Psi_0~ d\lam\\
  \nonumber
  &+&\frac 1{2\pi i}\int\limits_\Gamma e^{\lam t}h(i\lam)
  \Big[\cR _0\cV \cR _0\cV \cR (\lam-0)
  -\cR _0\cV \cR _0\cV \cR (\lam+0)\Big]\Psi_0~ d\lam\\
  &=&\Psi_{h1}(t)+\Psi_{h2}(t)+\Psi_{h3}(t),~~~~~~t\in\R \nonumber
\eeqn
We analyze each term $\Psi_{hk}$, $k=1,2,3$ separately.
\\
{\it Step i)}
The first term
$\Psi_{h1}(t)=\cG_h(t)\Psi_0$ by  (\ref{Gh}). Hence,
Theorem \ref{TD} implies that
\be\la{lins1}
  \Vert \Psi_{h1}(t)\Vert_{\cF _{-\si}}\le
 C(1+|t|)^{-3/2}\Vert \Psi_0\Vert_{\cF _\si},
  \quad t\in\R,\quad \sigma>5/2
\ee
{\it Step ii)}
Now we consider the second term $\Psi_{h2}(t)$.
Denote  $h_1(\om)=\sqrt{h(\om)}$ 
(we can assume that $h(\om)\ge 0$ and $h_1\in\C_0^\infty(\R)$).
We set
$$
\Phi_{h1}=\fr 1{2\pi i}\int\limits_\Gamma e^{\lam t}h_1(i\lam)
\Big[\cR _0(\lam-0)-\cR _0(\lam+0)\Big]\Psi_0~ d\lam
$$
It is obvious that for $\Phi_{h1}$ the inequality (\ref{lins1}) also holds.
Namely,
$$
  \Vert \Phi_{h1}(t)\Vert_{\cF _{-\si}}
 \le C(1+|t|)^{-3/2}\Vert \Psi_0\Vert_{\cF _\si},
  \quad t\in\R,\quad \sigma>5/2
$$
Further, the second term $\Psi_{h2}(t)$ can be written as a convolution.
\begin{lemma}(cf.  \ci[Lemma 3.11]{1dkg})
The convolution representation holds
\be\la{P2}
  \Psi_{h2}(t)=
  \int\limits_0^t \cG_{h1}(t-\tau)\cV \Phi_{h1}(\tau)~d\tau,~~~~t\in\R
\ee
where the integral converges in $\cF_{-\si}$ with $\si>5/2$.
\end{lemma}
Applying   Theorem \ref{TD} with $h_1$ instead of $h$ to the integrand
in (\re{P2}), we obtain that
$$
  \Vert \cG_{h1}(t-\tau)\cV \Phi_{h1}(\tau)\Vert_{\cF _{-\si}}
  \le \ds\fr{C\Vert \cV \Phi_{h1}(\tau)\Vert_{\cF _{\si'}}}{(1+|t-\tau|)^{3/2}}
  \le\ds\fr{C\Vert\Phi_{h1}(\tau)\Vert_{\cF _{\si'-\beta}}}{(1+|t-\tau|)^{3/2}}
  \le\ds\fr{C\Vert \Psi_0\Vert_{\cF _\si}}{(1+|t-\tau|)^{3/2}(1+|\tau|)^{3/2}}
$$
where $\si'\in (5/2,\beta-5/2)$.
Therefore, integrating here in $\tau$, we obtain by (\re{P2}) that
$$
  \Vert \Psi_{h2}(t)\Vert_{\cF _{-\si}}\le C(1+|t|)^{-3/2}\Vert \Psi_0\Vert_{\cF _\si},
  \quad t\in\R,\quad \sigma>5/2
$$
{\it Step iii)}
Let us rewrite the last term $\Psi_{h3}(t)$ as
$$
\Psi_{h3}(t)=\frac 1{2\pi i}\int\limits_{\Gamma}e^{\lam t}h(i\lam)
{\cal N}(\lam)\Psi_0~ d\lam,
$$
where ${\cal N}(\lam):={\cal M}(\lam-0)-{\cal M}(\lam+0)$ for
$\lam\in\Ga$, and
$$
  {\cal M}(\lam\pm 0):=\cR _0(\lam\pm 0)\cV
  \cR _0(\lam\pm 0)\cV \cR (\lam\pm 0)=
  \cR _0(\lam\pm 0){\cal W}(\lam\pm 0) \cR (\lam\pm 0),\quad
  \lam\in\Ga
$$
The asymptotics (\ref{bR0}), (\ref{bR}) and (\ref{bRV})
for $\cR_0^{(k)}(\lam\pm 0)$,  $\cR^{(k)}(\lam\pm 0)$ 
and ${\cal W}^{(k)}(\lam\pm 0)$ imply
\begin{lemma}\la{bM} (cf.\ci[Lemma 3.12]{1dkg}) For $k=0,1,2$  the asymptotics hold
$$
  \Vert{\cal M}^{(k)}(\lam\pm 0)\Vert_{\cL (\cF _{\sigma},\cF _{-\sigma})}
  ={\cal O}(|\lam|^{-2}),\quad |\lam|\to\infty,\quad \lam\in\Ga,\quad\si>1/2+k
$$
\end{lemma}
Finally, we prove the decay of $\Psi_{h3}(t)$.
By Lemma \ref{bM}
$$
(h {\cal N})''\in L^1((-i\infty,-\mu-i\ve)\cup(\mu+i\ve,i\infty);
\cL(\cF _{\si},\cF _{-\si}))
$$
with $\si>5/2$. Hence, two times partial integration implies that
$$
\Vert \Psi_{h3}(t)\Vert_{\cF _{-\si}}
 \le C(1+|t|)^{-2}\Vert \Psi_0\Vert_{\cF _\si},
  \quad t\in\R
$$
This completes the proof of Theorem \ref{main}.
\end{proof}
\begin{cor}
The asymptotics (\re{full1}) imply (\re{full}) with the projection
$$
\cP_c=1-\cP_d,\quad\cP_d=\sum_{\om_j\in\Si}P_j
$$
\end{cor}
\appendix
\section{Proof of Lemma \ref{l1}}

Formulas (\re{AX})- (\ref{J0}) imply
$$
\cG_v(z,t)=\tilde \cG_b(z,t)+{\tilde{\cal G}_r}(z,t)
$$
where
$$
\tilde \cG_b(z,t)\!=\!\frac
{\theta(t\!-\!|z\!-\!vt|)}{\sqrt{2m\pi}} \!\left(\ba{cc}\!\!\!\!
-\ds\fr{mt\sin(m\sqrt{t^2-(z-vt)^2}-\frac{\pi}4)}{\sqrt[4]{(t^2-(z-vt)^2)^3}}
&\ds\fr{\cos(m\sqrt{t^2-(z-vt)^2}-\frac{\pi}4)}{\sqrt[4]{t^2-(z-vt)^2}} \\\\
\!\!\!\!
-\ds\fr{m^2t^2\cos(m\sqrt{t^2-(z-vt)^2}-\frac{\pi}4)}{\sqrt[4]{(t^2-(z-vt)^2)^5}}
&-\ds\fr{mt\sin(m\sqrt{t^2-(z-vt)^2}-\frac{\pi}4)}{\sqrt[4]{(t^2-(z-vt)^2)^3}}
\ea \!\!\!\!\right)
$$
For $\ve\in (|v|,1)$ and $|z|\le(\ve-|v|)t$ we have $|z-vt|\le\ve t$. Hence 
$$
|\partial^{k}_z{\tilde{\cal G}}_r(z,t)|\le C(\ve )t^{-3/2},
\quad|z|\le(\ve-|v|) t,\quad k=0,1
$$
by known asymptotics of the Bessel function (see \ci{W}, p.195).
It remains to prove the bounds of type (\re{G2R}) for the
difference $Q(z,t)=\tilde \cG_b(z,t)-\cG_b(z,t)$.
Let us consider the entry $Q^{12}(t,z)$: 
$$
Q^{12}(t,z)=\frac {1}{\sqrt{2\pi m}}
\Big[\fr{\cos(m\sqrt{t^2-(z-vt)^2}-\frac{\pi}4)}{\sqrt[4]{t^2-(z-vt)^2}}
-\frac{\cos(m(\fr{t}{\ga}+\ga vz)-\frac{\pi}4)}{\sqrt{t/\ga}}\Big]
$$
For $|z|\le(\ve-|v|) t$ we have
\beqn\nonumber
&&\Big|\frac 1{\sqrt[4]{t^2-(z-vt)^2}}-\frac
1{\sqrt{t/\ga}}\Big|\\
\nonumber
&&=\frac{|z^2-2vtz|}
{\sqrt[4]{t^2-(z-vt)^2}\sqrt{t/\ga}\big(\sqrt[4]{t^2-(z-vt)^2}+\sqrt{t/\ga}\big)
\big(\sqrt{t^2-(z-vt)^2}+t/\ga\big)}\le\frac {C(\ve) |z|}{t\sqrt t}
\eeqn
Further,
$$
\Big|\cos\Big(m\sqrt{t^2\!-\!(z\!-\!vt)^2}-\frac{\pi}4\Big)
-\cos\Big(\fr m\ga(t+\ga^2 vz)-\frac{\pi}4\Big)\Big| \le
2\Big|\sin\Big(\frac m2(\sqrt{t^2\!-\!(z\!-\!vt)^2}-\fr{t+\ga^2 vz}{\ga}\Big)\Big|
$$
$$
\le C\Big|\sqrt{t^2-(z-vt)^2}-(t+\ga^2 vz)/\ga\Big|
\le C\fr{z^2(1+\ga^2v^2)}{|\sqrt{t^2-(z-vt)^2}+(t+\ga^2vz)/\ga|}
\le\frac {C(\ve) z^2}{t}
$$
since $\ga^2|v||z|\le(1-|v|)t/(1-v^2)\le t/(1+|v|)\le t$.
Hence, 
\be\la{a} 
|Q^{12}(t,z)|\le C(\ve)(1+z^2)t^{-3/2},~~~|z|\le(\ve-|v|) t 
\ee
 Differentiating
$Q^{12}(t,z)$, we obtain for $|z|\le(\ve- |v|) t$
$$
\partial_z Q^{12}(t,z)=\fr{z-vt}{\sqrt{2\pi m}}
~\frac{\cos(m\sqrt{t^2-(z-vt)^2}-\frac{\pi}4)}{2\sqrt[4]{(t^2-(z-vt)^2)^5}}
+\sqrt{\fr{m}{2\pi}}\fr{z\sin(m\sqrt{t^2-(z-vt)^2}-\frac{\pi}4)}{\sqrt[4]{(t^2-(z-vt)^2)^3}}
$$
$$
+\sqrt{\fr{m}{2\pi}}vt\Big[\fr{-\sin(m\sqrt{t^2-(z-vt)^2}-\frac{\pi}4)}
{\sqrt[4]{(t^2-(z-vt)^2)^3}}+\fr{\sin(m(\fr{t}{\ga}+\ga
vz)-\frac{\pi}4)} {\sqrt{(t/\ga)^3}}\Big]
$$
Hence, by the arguments above, 
\be\la{b}
 |\partial_z Q^{12}(t,z)|\le  C(\ve)(1+z^2)\,t^{-3/2},\quad |z|\le(\ve-|v|) t 
\ee
Other entries  $Q^{ij}(t,z)$ also admit the estimates of type
(\ref{a}) and (\ref{b}). Hence, the lemma follows since ${\cal
G}_r(t)=\tilde{\cal G}_r(t)+Q(t,z)$.

\end{document}